\documentclass[journal,10pt]{IEEEtran}
\usepackage{graphics} % for pdf, bitmapped graphics files
\usepackage{epsfig} % for postscript graphics files
\usepackage{amsmath} % assumes amsmath package installed
\usepackage{cite}

\usepackage{algorithm}
\usepackage{algorithmic}

\usepackage{amssymb}  % assumes amsmath package installed
\usepackage{booktabs}
\usepackage{tabularx}
\usepackage{makecell}
\usepackage[dvipdfm,
            CJKbookmarks=true,
            bookmarksnumbered=true,
            bookmarksopen=true,
            colorlinks=true,
            citecolor=black,
            linkcolor=black,
            anchorcolor=green,
            urlcolor=blue
            ]{hyperref}

\IEEEoverridecommandlockouts                              % This command is only
\newtheorem{theorem}{Theorem}

\title{A New Approach to Linear/Nonlinear Distributed Fusion Estimation Problem}
\author{Bo Chen,~\IEEEmembership{Member,~IEEE},~Guoqiang Hu,~\IEEEmembership{Member,~IEEE},~Daniel W.C. Ho,~\IEEEmembership{Fellow,~IEEE},~Li Yu,~\IEEEmembership{Member,~IEEE}% <-this % stops a space

\thanks{B. Chen and G. Hu are with the School of Electrical and Electronic Engineering, Nanyang Technological University, 639798 Singapore
(email: bchen@aliyun.com; gqhu@ntu.edu.sg).}
\thanks{D. W. C. Ho is with the Department of Mathematics, City University of Hong Kong, Hong Kong, 999077.
(email: madaniel@cityu.edu.hk).}
\thanks{L. Yu is with the College of Information Engineering, Zhejiang University of Technology, HangZhou 310023, China
(email: lyu@zjut.edu.cn).}}

\begin{document}

\markboth{Final Version}
{Shell \MakeLowercase{\textit{et al.}}: Bare Demo of IEEEtran.cls for Journals}
\maketitle

%%%%%%%%%%%%%%%%%%%%%%%%%%%%%%%%%%%%%%%%%%%%%%%%%%%%%%%%%%%%%%%%%%%%%%%%%%%%%%%%
\begin{abstract}
In this paper, we study the distributed fusion estimation problem for linear time-varying systems and nonlinear systems with bounded noises, where the addressed noises do not provide any statistical information, and are unknown but bounded. When considering linear time-varying fusion systems with bounded noises, a new local Kalman-like estimator is designed such that the square error of the estimator is bounded as time goes to $\infty$. A novel constructive method is proposed to find an upper bound of fusion estimation error, then a convex optimization problem on the design of an optimal weighting fusion criterion is established in terms of linear matrix inequalities, which can be solved by standard software packages. Furthermore, according to the design method of linear time-varying fusion systems, each local nonlinear estimator is derived for nonlinear systems with bounded noises by using Taylor series expansion, and a corresponding distributed fusion criterion is obtained by solving a convex optimization problem. Finally, target tracking system and localization of a mobile robot are given to show the advantages and effectiveness of the proposed methods.
\end{abstract}


%%%%%%%%%%%%%%%%%%%%%%%%%%%%%%%%%%%%%%%%%%%%%%%%%%%%%%%%%%%%%%%%%%%%%%%%%%%%%%%%

%%%%%%%%%%%%%%%%%%%%%%%%%%%%%%%%%%%%%%%%%%%%%%%%%%%%%%%%%%%%%%%%%%%%%%%%%%%%%%%%
\section{Introduction}
Multi-sensor fusion estimation has been one of the most important focuses in the area of information fusion during the past two decades. Since estimation performance and reliability can be improved by the redundant information from multiple sensors, different fusion estimation methods have been found in many application fields such as target tracking and localization, guidance and navigation \cite{c1}, fault detection \cite{c2}, sensor networks \cite{ac2} and cyber-physical systems \cite{c3}. Generally, there exist two kinds of fusion estimation structures: centralized fusion structure and distributed fusion structure. Under the centralized fusion structure, all measurement data from different sensors are communicated to the fusion center (FC), and the design of centralized fusion estimator is equivalent to that of state estimator with single sensor. Though the centralized fusion estimation can provide optimal estimation performance, it has poor robustness and reliability when there are faulty sensors and FC. Under the distributed fusion structure, the measurements are first used to derive the local estimates at each sensor, and then these local estimates are sent to the FC to yield optimal or suboptimal state estimate by designing certain fusion criteria. Compared with the centralized fusion structure, the distributed fusion structure is generally more robust, reliable, and fault tolerant \cite{c4}. Therefore, most of existing works are focused on how to design the distributed fusion estimation algorithms.

Generally, the physical processes are modeled by linear or nonlinear dynamical systems, while the disturbance noises in multi-sensor fusion systems are considered as Gaussian or non-Gaussian disturbances. When considering the Gaussian white noise with known covariances, there mainly exist three different distributed fusion estimation methods: i) Optimal distributed fusion estimation strategies \cite{c5,c6,c7}; ii) Suboptimal distributed weighted fusion estimation methods \cite{c8,c9,c10}; iii) Suboptimal distributed covariance intersection fusion estimation methods \cite{c11,c12,c13,c14}. Notice that the assumption of Gaussian white noises may not be satisfied in practical systems, particularly, the accurate covariances may not be obtained in practical applications. To overcome this drawback, the energy-bounded noises, which do not require any statistical property of noises, have been considered for multi-sensor fusion systems, and different distributed ${H_\infty }$ fusion estimation methods were developed in \cite{c15,c16,c17}. Subsequently, when simultaneously considering the energy-bounded noises and Gaussian white noises with unknown covariances, the distributed mixed ${H_2}/{H_\infty }$ fusion estimation algorithms have been developed in \cite{c18,c19} for a class of networked fusion systems. Though the conditions of disturbance noises have been relaxed in \cite{c15,c16,c17,c18,c19}, the addressed fusion systems in \cite{c15,c16,c17,c18,c19} were time-invariant. Moreover, the energy-bounded noise $w(t)$ in the ${H_\infty }$ fusion framework must satisfy ${\lim _{t \to \infty }}w(t) = 0$, which may not be true in some practical systems (e.g., sensor's measurement noise generated from the external environment always exists). To make up for these shortages, a novel fusion method was developed in \cite{c20} to solve the networked fusion estimation problem, where a general framework was proposed in \cite{c20} to deal with state estimation problem under bounded noises. However, how to find the closest upper bounds and the most suitable optimization objective were not well solved in \cite{c20}, which still remains challenging.

It should be pointed out that the works in \cite{c5,c6,c7,c8,c9,c10,c11,c12,c13,c14,c15,c16,c17,c18,c19,c20} were concerned with the fusion estimation problem of linear systems. When considering the nonlinear systems with Gaussian white noises, the distributed fusion estimation algorithm was developed in \cite{c25} by using the extended Kalman filter (EKF) \cite{c26}, while the unscented information fusion filtering algorithm was derived in \cite{c27} by using the unscented Kalman filter (UKF) \cite{c28,c29}. Meanwhile, the fifth-degree ensemble iterated cubature square-root information filter was introduced in \cite{c31} to design nonlinear fusion estimation algorithm, while the support vector regression methodology was proposed in \cite{c32} to design nonlinear fusion strategy. Recently, different sequential fusion estimation methods for nonlinear systems were presented in \cite{c30,c33} based on the UKF. Notice that the fusion estimation methods in \cite{c25,c27,c31,c32,c30,c33} assumed that the system disturbances must be Gaussian white noises with known covariance. On the other hand, when considering non-Gaussian noises in nonlinear fusion estimation framework, the consensus and Rao-Blackwellized fusion particle filtering algorithms were designed in \cite{c34,c35}, where the probability density function was required to be known in advance. Meanwhile, the modified Kalman filtering methods in \cite{c40,c41} may also be used to solve nonlinear fusion estimation problem, where the statistical information of noises are required to be known in advance. As mentioned before, the disturbance noises in practical systems are always bounded, and the statistical property of noises is difficult to be accurately obtained in practical applications. However, the above methods are not applicable to this case, and few results are focused on the distributed fusion estimation problem for nonlinear systems with bounded noises.

Motivated by the above analysis, we shall study the distributed fusion estimation problem for linear time-varying systems and nonlinear systems with bounded noises. The main contributions of this paper are summarized as follows: i) For the linear time-varying systems, a new \emph{stable} local Kalman-like estimator, which is different from the estimator structure in \cite{c20}, is obtained by solving a convex optimization problem at each time step. Subsequently, by constructing a new upper bound of fusion estimation error, an optimal distributed fusion criterion is designed by solving a class of convex optimization problems; ii) Linearizing the nonlinear function using Taylor series expansion, the general nonlinear system reduces to linear time-varying systems, and the linearized errors can be viewed as bounded noises. Under this case, according to the design method of linear time-varying fusion systems, each local nonlinear estimator and a distributed fusion criterion are designed by establishing different convex optimization problems; iii) The proposed fusion estimation method in this paper does not require any statistical information of noises as compared with the classical Kalman fusion estimation methods, and this method can also be applicable to linear time-varying systems and nonlinear systems as compared with the existing $H_\infty $ fusion estimation methods. Notice that the solutions to the convex optimization problems in this paper can be directly obtained by the standard software packages. In the simulations, the advantages of the linear estimation method in this paper are shown by comparing with the state estimator in \cite{c20}, the Kalman filter \cite{c21},  and the ${H_\infty }$ filter in \cite{c22}, while the advantages of the nonlinear estimation method in this paper are shown by comparing with the classical EKF \cite{c26} and UKF \cite{c28,c29}.

\emph{Notations}: The superscript ``${\rm{T}}$'' represents the transpose, while ``$I$'' represents the identity matrix with appropriate dimension. $X>(<)0$ denotes a positive-definite (negative-definite) matrix, and ${\rm{diag}\{\cdot\}}$ stands for a block diagonal matrix. ${\lambda _{\max }}( \cdot )$ means the maximum eigenvalue of the corresponding matrix, while $||A|{|_2}$ is the 2-norm of matrix $A$. The symmetric terms in a symmetric matrix are denoted by ``$\ast$'', and ${\rm{col}}\{ {a_1}, \cdots ,{a_n}\}$ means a column vector whose elements are ${a_1}, \cdots ,{a_n}$. Moreover, if~${\tau _1} > {\tau _2},$~it will be specified that $\prod\nolimits_{\tau  = {\tau _1}}^{{\tau _2}} {F(\tau )}  = {I_m}$, where $F(\tau ) \in \mathbb{R}^{m \times m}$ represents a matrix function with respect to the variable $\tau$.

%%%%%%%%%%%%%%%%%%%%%%%%%%%%%%%%%%%%%%%%%%%%%%%%%%%%%%%%%%%%%%%%%%%%%%%%%%%
\section{Problem Statement}
Consider a nonlinear system described by the following state-space model:
\begin{eqnarray}
\begin{array}{c}
{\rm{x}}(t + 1) = {\rm{f}}({\rm{x}}(t)) + B(t)w(t)\;\;\;\;\;\;\;\;\;\;\;\;\;\;\;\;
\end{array}
\label {eq:1}\\
%\end{eqnarray}
%\begin{eqnarray}
{y_i}(t) = {{\rm{g}}_i}({\rm{x}}(t)) + {B_i}(t)v_i(t)(i = 1,2, \cdots ,L)
\label {eq:2}
\end{eqnarray}
where $x(t) \in {{\rm{R}}^n}$ is the system state, ${y_i}(t) \in {{\rm{R}}^{{q_i}}}$ is the measured output from sensor $i$, and $L$ is the number of sensors. ${\rm{f}}(x(t)) \in {{\rm{R}}^{n \times 1}}$ and ${{\rm{g}}_i}(x(t)) \in {{\rm{R}}^{{q_i} \times 1}}$ are nonlinear vector functions that are assumed to be continuously differentiable, while $B(t)$ and ${B_i}(t)(i = 1,2, \cdots ,L)$ are time-varying bounded matrices with appropriate dimensions. $w(t)$ and $v_i(t)$ are \emph{bounded noises}, i.e.,
\begin{eqnarray}
{w^{\rm{T}}}(t)w(t) \le {\delta _w},{v_i^{\rm{T}}}(t)v_i(t) \le {\delta _{v_i}}
\label {eq:3}
\end{eqnarray}
where ${\delta _w}$ and ${\delta _{v_i}}$ are \emph{unknown}. At each sensor, based on the measurements $\{ {y_i}(1),{y_i}(2), \cdots ,{y_i}(t)\}$, the local state estimate (LSE) ${{{\rm{\hat x}}}_i}(t)$ for nonlinear systems (\ref{eq:1}--\ref{eq:2}) is given by:
\begin{eqnarray}
\left\{ \begin{array}{l}
 {\rm{\hat x}}_i^{\rm{p}}(t) = {\rm{f}}({{{\rm{\hat x}}}_i}(t - 1)) \\
 {{{\rm{\hat x}}}_i}(t) = {\rm{\hat x}}_i^{\rm{p}}(t) + {\rm K}_i^N(t)[{y_i}(t) - {{\rm{g}}_i}({\rm{\hat x}}_i^{\rm{p}}(t))] \\
 \end{array} \right.
\label {eq:4}
\end{eqnarray}
where ${\rm{\hat x}}_i^{\rm{p}}(t)$ denotes one-step prediction, and an optimal gain ${\rm K}_i^N(t)$ is to be designed in Section III.

When the nonlinear systems (\ref{eq:1}--\ref{eq:2}) are reduced to the following linear discrete time-varying systems:
\begin{eqnarray}
\begin{array}{c}
{\rm{x}}(t + 1) = A(t){\rm{x}}(t) + B(t)w(t)\;\;\;\;\;\;\;\;\;\;\;\;\;\;\;\;
\end{array}
\label {eq:5}\\
%\end{eqnarray}
%\begin{eqnarray}
{y_i}(t) = {C_i}(t){\rm{x}}(t) + {B_i}(t)v_i(t)(i = 1,2, \cdots ,L)
\label {eq:6}
\end{eqnarray}
where $A(t)$ and ${C_i}(t)$ are time-varying matrices with appropriate dimensions. Then, each LSE ${{{\rm{\hat x}}}_i}(t)$ for linear systems (\ref{eq:5}--\ref{eq:6}) is given by the Kalman-like structure:
\begin{eqnarray}
\begin{array}{l}
 {{{\rm{\hat x}}}_i}(t) = A(t - 1){{{\rm{\hat x}}}_i}(t - 1) \\
 \;\;\;\;\;\;\;\;\;\;\;\;\;\;\;\;\;\;\;\;\;\; + {{\rm K}_i^L}(t)[{y_i}(t) - {C_i}(t)A(t - 1){{{\rm{\hat x}}}_i}(t - 1)] \\
 \end{array}
\label {eq:7}
\end{eqnarray}
where an optimal gain ${{\rm K}_i^L}(t)$ is to be designed in Section III.

Subsequently, based on the LSEs (\ref{eq:4}) or (\ref{eq:7}), the distributed fusion estimate (DFE) of ${\rm{x}}(t)$ is given by:
\begin{eqnarray}
{\rm{\hat x}}(t) = \sum\nolimits_{i = 1}^L {{\Omega _i}(t){{{\rm{\hat x}}}_i}(t)}
\label {eq:8}
\end{eqnarray}
where $\sum\nolimits_{i = 1}^L {{\Omega _i}(t)}  = I$, and optimal weighting matrices ${\Omega _1}(t), \cdots ,{\Omega _L}(t)$ will be designed in Section III. Consequently, the problems to be solved in this paper are described as follows:
\begin{itemize}
\item The first aim is to design optimal gains ${\rm K}_i^N(t)$ in (\ref{eq:4}) and ${\rm K}_i^L(t)$ in (\ref{eq:7}) such that an upper bound of the square error (SE) of the corresponding LSE ${{{\rm{\hat x}}}_i}(t)$ is minimal at each time, and the SE of ${{{\rm{\hat x}}}_i}(t)$ for linear systems is bounded as $t$ goes to $\infty$;
\item The second aim is to design optimal weighting matrices ${\Omega _1}(t), \cdots ,{\Omega _L}(t)$ in (\ref{eq:8}) such that an upper bound of the SE of the DFE ${\rm{\hat x}}(t)$ is minimum at each time.
\end{itemize}

\textbf{Remark 1}. When considering the linear time-varying systems with bounded noises, the LSE in \cite{c20} was given by
\begin{eqnarray}
\begin{array}{l}
 {{{\rm{\hat x}}}_i}(t) = A(t - 1){{{\rm{\hat x}}}_i}(t - 1) \\
 \;\;\;\;\;\;\;\;\;\;\; + {{\rm{K}}_i}(t)({y_i}(t - 1) - {C_i}(t - 1){{{\rm{\hat x}}}_i}(t - 1)) \\
 \end{array}
\label {eq:r1}
\end{eqnarray}
where ${{\rm{K}}_i}(t)$ is the estimator gain of (\ref{eq:r1}). The difference of the estimator structures between the LSE (\ref{eq:7}) and the LSE (\ref{eq:r1}) is that the LSE (\ref{eq:r1}) in \cite{c20} was designed based on the measurements $\{ {y_i}(1), \cdots ,{y_i}(t - 1)\} $, while the LSE (\ref{eq:4}) is designed based on the measurements $\{ {y_i}(1), \cdots ,{y_i}(t - 1),{y_i}(t)\}$. Notice that ${y_i}(t)$ can provide important information of $x(t)$, but the design of the LSE (\ref{eq:r1}) did not use the measurement ${y_i}(t)$. Thus, more available information on the state $x(t)$ is used in this paper to design the estimator. In this sense, under the same criterion of performance assessment, the estimation performance of the LSE (\ref{eq:7}) can be improved as compared with the LSE in \cite{c20}.

\textbf{Remark 2}. Compared with the Kalman fusion estimation algorithms in \cite{ac2,c3,c4,c5,c6,c7,c8,c9,c10,c11,c12,c13,c14,c25,c27,c31,c32,c30,c33}, the proposed fusion estimation algorithm in this paper does not require any statistical information of the noises. Since the Gaussian white noises are always bounded in a practical system, the fusion estimation algorithm in this paper is also applicable to the case of Gaussian white noises.

\textbf{Remark 3}. Compared with the ${H_\infty }$ fusion estimation algorithms in \cite{c15,c16,c17,c18,c19}, the proposed fusion estimation algorithm in this paper does not require that the noise $w(t)$ (or $v_i(t)$) is energy-bounded (i.e., ${\lim _{t \to \infty }}w(t)=0$), and is applicable to time-varying systems and nonlinear systems. Notice that the energy-bounded noise can be viewed as a special case of bounded noises. Moreover, the nonlinear fusion estimation methods based on Taylor series expansion cannot be obtained from the similar ideas in \cite{c15,c16,c17,c18,c19}, because the fusion estimation methods in \cite{c15,c16,c17,c18,c19} only dealt with linear time-invariant systems, but the linearized systems must be time-varying.

\section{Main Results}
\subsection{DFE Design for Linear Time-Varying Systems}
In this subsection, an optimal local estimator gain ${{\rm K}_i^L}(t)$ in (\ref{eq:7}) and optimal weighting matrices ${\Omega _1}(t), \cdots ,{\Omega _L}(t)$ in (\ref{eq:8}) for linear systems (\ref{eq:5}--\ref{eq:6}) will be presented in Theorem 1. Before deriving the result of Theorem 1, let us define
\begin{eqnarray}
\left\{ \begin{array}{l}
 {\rm{G}}_{{{\rm{K}}_i}}^L(t) \buildrel \Delta \over = I - {\rm{K}}_i^L(t){C_i}(t) \\
 B_{{f_i}}^L(t) \buildrel \Delta \over = [{\rm{G}}_{{{\rm{K}}_i}}^L(t)B(t - 1)\:\: - {\rm{K}}_i^L(t){B_i}(t)] \\
 \bar B_{{f_i}}^L(t) \buildrel \Delta \over = [{\rm{G}}_{{{\rm{K}}_i}}^L(t)B(t - 1)\;0 \cdots \: - {\rm{K}}_i^L(t){B_i}(t) \cdots 0] \\
 A_F^L(t) \buildrel \Delta \over = {\rm{diag}}\{ {\rm{G}}_{{{\rm{K}}_1}}^L(t)A(t - 1), \cdots ,{\rm{G}}_{{{\rm{K}}_L}}^L(t)A(t - 1)\}  \\
 B_F^L(t) \buildrel \Delta \over = {\rm{col}}\{ \bar B_{{f_1}}^L(t), \cdots ,\bar B_{{f_L}}^L(t)\}  \\
 \end{array} \right.
\label {eq:10}
\end{eqnarray}
\begin{theorem}
An optimal estimator gain ${{\rm K}_i^L}(t)$ in (\ref{eq:7}) can be obtained by solving the following convex optimization problem:
\begin{eqnarray}
\begin{array}{l}
 \mathop {\min }\limits_{{\vartheta _i}(t) > 0,{P_i}(t) > 0,{\Theta _i}(t) > 0,{{\rm K}_i^L}(t)} {\rm{Tr}}\{ {\Theta _i}(t)\}  \\
 {\rm{s}}{\rm{.t}}{\rm{.}}:\left\{ \begin{array}{l}
 \left[ {\begin{array}{*{20}{c}}
   { - I} & {{{\rm{G}}_{{{\rm K}_i}}^L}(t)A(t - 1)} & {{B_{{f_i}}^L}(t)}  \\
    *  & { - {P_i}(t)} & 0  \\
    *  &  *  & { - {\Theta _i}(t)}  \\
\end{array}} \right] < 0 \\
 {P_i}(t) - {\vartheta _i}(t)I < 0 \\
 {\vartheta _i}(t) < 1 \\
 \end{array} \right. \\
 \end{array}
\label {eq:11}
\end{eqnarray}
where ${\rm{G}}_{{{\rm{K}}_i}}^L(t)$ and $B_{{f_i}}^L(t)$ are defined in (\ref{eq:10}). Under this case, the SE of ${{{\rm{\hat x}}}_i}(t)$ will be bounded, i.e., there must exist a positive scalar ${p_i} > 0$ such that
\begin{eqnarray}
\mathop {\lim }\limits_{t \to \infty } {\rm{e}}_i^{\rm{T}}(t){{\rm{e}}_i}(t) < {p_i}
\label {eq:12}
\end{eqnarray}
Moreover, a group of optimal weighting matrices ${\Omega _1}(t), \cdots ,{\Omega _L}(t)$ in (\ref{eq:8}) for linear systems (\ref{eq:5}--\ref{eq:6}) can be obtained by solving the following convex optimization problem:
\begin{eqnarray}
\begin{array}{l}
 \mathop {\min }\limits_{\Omega (t),\Upsilon (t),P(t) > 0,\Theta (t) > 0} {\rm{Tr}}\{ P(t)\}  + {\rm{Tr}}\{ \Theta (t)\}  \\
 {\rm{s}}{\rm{.t}}{\rm{.}}:\left[ {\begin{array}{*{20}{c}}
   { - I} & {\Omega (t){A_F^L}(t)} & {\Omega (t){B_F^L}(t)}  \\
    *  & { - P(t)} & { - \Upsilon (t)}  \\
    *  &  *  & { - \Theta (t)}  \\
\end{array}} \right] < 0 \\
 \end{array}
\label {eq:13}
\end{eqnarray}
where $\Omega (t) \buildrel \Delta \over = [{\Omega _1}(t), \cdots ,{\Omega _{L - 1}}(t),I - \sum\nolimits_{i = 1}^{L - 1} {{\Omega _i}(t)} ]$, while ${{A_F^L}(t)}$ and ${{B_F^L}(t)}$ are defined in (\ref{eq:10}).
\end{theorem}

\begin{proof}
Define ${{\rm{e}}_i}(t) \buildrel \Delta \over = {\rm {x(t)}} - {{{\rm{\hat x}}}_i}(t)$ and $\xi_i (t - 1) \buildrel \Delta \over = {\rm{col}}\{ w(t - 1),v_i(t)\}$. Then, the estimation error ${{\rm{e}}_i}(t)$ is given by:
\begin{eqnarray}
{{\rm{e}}_i}(t) = {\rm{G}}_{{{\rm K}_i}}^L(t)A(t - 1){{\rm{e}}_i}(t - 1) + B_{{f_i}}^L(t)\xi_i (t - 1)
\label {eq:a6}
\end{eqnarray}
where ${\rm{G}}_{{{\rm{K}}_i}}^L(t)$ and $B_{{f_i}}^L(t)$ are defined in (\ref{eq:10}).

Next, the following performance index is introduced:
\begin{eqnarray}
\begin{array}{l}
 J_i(t) \buildrel \Delta \over = {\rm{e}}_i^{\rm{T}}(t){{\rm{e}}_i}(t) - {\rm{e}}_i^{\rm{T}}(t - 1){P_i}(t){{\rm{e}}_i}(t - 1) \\
  \;\;\;\;\;\;\;\;\;\;\;- {\xi_i ^{\rm{T}}}(t - 1){\Theta _i}(t)\xi_i (t - 1) \\
 \end{array}
\label {eq:14}
\end{eqnarray}
where ${P_i}(t) > 0$ and ${\Theta _i}(t)>0$. Then, it follows from (\ref{eq:a6}) that
\begin{eqnarray}
{J_i}(t) = {\left[ {\begin{array}{*{20}{c}}
   {{{\rm{e}}_i}(t - 1)}  \\
   {\xi_i (t - 1)}  \\
\end{array}} \right]^{\rm{T}}}\underbrace {\left[ {\begin{array}{*{20}{c}}
   {{Z_{i1}}(t)} & {{Z_{i2}}(t)}  \\
    *  & {{Z_{i3}}(t)}  \\
\end{array}} \right]}_{{Z_i}(t)}\left[ {\begin{array}{*{20}{c}}
   {{{\rm{e}}_i}(t - 1)}  \\
   {\xi_i (t - 1)}  \\
\end{array}} \right]
\label {eq:15}
\end{eqnarray}
where ${Z_{i1}}(t) \buildrel \Delta \over = {A^{\rm{T}}}(t - 1){{\rm{[G}}_{{{\rm K}_i}}^L(t)]^{\rm{T}}}{\rm{G}}_{{{\rm K}_i}}^L(t)A(t - 1) - {P_i}(t)$, ${Z_{i2}}(t) \buildrel \Delta \over = {A^{\rm{T}}}(t - 1){{\rm{[G}}_{{{\rm K}_i}}^L(t)]^{\rm{T}}}{B_{{f_i}}^L}(t)$ and ${Z_{i3}}(t) \buildrel \Delta \over = {[B_{{f_i}}^L(t)]^{\rm{T}}}B_{{f_i}}^L(t) - {\Theta _i}(t)$. According to the Schur complement lemma \cite{c23}, the first inequality in (\ref{eq:11}) is equivalent to ${Z_i}(t) < 0$. This means that ${J_i}(t) < 0$ under the first inequality in (\ref{eq:11}), and thus one has
\begin{eqnarray}
\begin{array}{l}
 {\rm{e}}_i^{\rm{T}}(t){{\rm{e}}_i}(t) < {\rm{e}}_i^{\rm{T}}(t - 1){P_i}(t){{\rm{e}}_i}(t - 1) \\
 \;\;\;\;\;\;\;\;\;\;\;\;\;\;\;\;\;\;\;\;\;\;\;\;\;\;\;\;\;\;\;\;\;\; + {\xi_i ^{\rm{T}}}(t - 1){\Theta _i}(t)\xi_i (t - 1) \\
 \end{array}
\label {eq:17}
\end{eqnarray}
When the inequality ${P_i}(t) - {\vartheta _{i}}(t)I < 0$ in (\ref{eq:11}) holds, one has ${\lambda _{\max }}({P_i}(t)) < {\vartheta _{i}}(t)$. Then, combining (\ref{eq:17}) yields that
\begin{eqnarray}
\begin{array}{l}
 {\rm{e}}_i^{\rm{T}}(t){{\rm{e}}_i}(t) < {\vartheta _i}(t){\rm{e}}_i^{\rm{T}}(t - 1){{\rm{e}}_i}(t - 1) \\
 \;\;\;\;\;\;\;\;\;\;\;\;\;\;\;\;\;\;\;\;\;\;\;\;\;\;\;\;\;\;\;\;\;\;\;\; + {\xi_i ^{\rm{T}}}(t - 1){\Theta _i}(t)\xi_i (t - 1) \\
 \end{array}
\label {eq:19}
\end{eqnarray}
Thus, it is derived from (\ref{eq:19}) that
\begin{eqnarray}
\begin{array}{l}
 {\rm{e}}_i^{\rm{T}}(t){{\rm{e}}_i}(t) < \left( {\prod\nolimits_{\kappa  = 0}^{t - 1} {{\vartheta _i}(t - \kappa )} } \right){\rm{e}}_i^{\rm{T}}(0){{\rm{e}}_i}(0) \\
 \;\;\;\;\;\;\;\;\;\;\;\;\;\;\;\; + \sum\nolimits_{\kappa  = 0}^{t - 1} {\left\{ {\left( {\prod\nolimits_{\tau  = 0}^{\kappa  - 1} {{\vartheta _i}(t - \tau )} } \right)} \right.}  \\
 \;\;\;\;\;\;\;\;\;\;\;\;\;\;\;\; \times {\xi_i ^{\rm{T}}}(t - \kappa  - 1){\Theta _i}(t - \kappa )\xi_i (t - \kappa  - 1)\}  \\
 \end{array}
\label {eq:20}
\end{eqnarray}
When the condition ``${\vartheta _i}(t) < 1$'' in (\ref{eq:11}) holds, one has
\begin{eqnarray}
\mathop {\lim }\limits_{t \to \infty } \prod\nolimits_{\kappa  = 0}^{t - 1} {{\vartheta _i}(t - \kappa )}  = 0,\mathop {\lim }\limits_{\kappa  \to \infty } \prod\nolimits_{\tau  = 0}^{\kappa  - 1} {{\vartheta _i}(t - \tau )}  = 0
\label {eq:21}
\end{eqnarray}
Then, it is concluded from (\ref{eq:20}--\ref{eq:21}) that $\mathop {\lim }\limits_{t \to \infty } {\rm{e}}_i^{\rm{T}}(t){{\rm{e}}_i}(t)$ is bounded.

Notice that ${\xi_i ^{\rm{T}}}(t - 1){\Theta _i}(t)\xi_i (t - 1) \le {\lambda _{\max }}(\xi_i (t - 1){\xi_i ^{\rm{T}}}(t - 1)){\rm{Tr}}\{ {\Theta _i}(t)\}$, and thus it follows from (\ref{eq:17}) that
\begin{eqnarray}
\begin{array}{l}
 {\rm{e}}_i^{\rm{T}}(t){{\rm{e}}_i}(t) < {\vartheta _i}(t){\rm{e}}_i^{\rm{T}}(t - 1){{\rm{e}}_i}(t - 1) \\
 \;\;\;\;\;\;\;\;\;\;\;\;\;\;\;\;\;\;\;\;\;\;\;\; + {\lambda _{\max }}(\xi_i (t - 1){\xi_i ^{\rm{T}}}(t - 1)){\rm{Tr}}\{ {\Theta _i}(t)\}  \\
 \end{array}
\label {eq:18}
\end{eqnarray}
In this case, the right term of (\ref{eq:18}) can be viewed as an upper bound of ${\rm{e}}_i^{\rm{T}}(t){{\rm{e}}_i}(t)$ at each time. Though the estimation error ${{\rm{e}}_i}(t)$ is generated from the the initial error ${{\rm{e}}_i}(0)$ and the bounded noises $\xi (0), \cdots ,\xi (t - 1)$, it is concluded from (\ref{eq:19}) and (\ref{eq:20}) that when the third condition in (\ref{eq:11}) holds, the estimation error is independent of the initial value. Based on the above consideration, ``$\min \;{\rm{Tr}}\{ {\Theta _i}(t)\}$'' is proposed to be the optimization objective when minimizing this upper bound at each time.

Define ${{\rm{e}}_0}(t) \buildrel \Delta \over = x(t) - {\rm{\hat x}}(t)$. Then, the fusion estimation error is calculated by:
\begin{eqnarray}
{{\rm{e}}_0}(t) = \sum\nolimits_{i = 1}^L {{\Omega _i}(t){{\rm{e}}_i}(t)}
\label {eq:a9}
\end{eqnarray}
where ${{{\rm{e}}_i}(t)}$ is determined by (\ref{eq:a6}). To design a group of optimal weighting matrices in (\ref{eq:8}), the following fusion system is constructed from the estimation error (\ref{eq:a6}) and the fusion estimation error (\ref{eq:a9}):
\begin{eqnarray}
\left\{ \begin{array}{l}
 {{\rm{e}}_F}(t) = {A_F^L}(t){{\rm{e}}_F}(t - 1) + {B_F^L}(t)\xi (t - 1) \\
 {{\rm{e}}_0}(t) = \Omega (t){{\rm{e}}_F}(t) \\
 \end{array} \right.
\label {eq:22}
\end{eqnarray}
where ${{\rm{e}}_F}(t) \buildrel \Delta \over = {\rm{col}}\{ {{\rm{e}}_1}(t), \cdots ,{{\rm{e}}_L}(t)\}$ and $\xi (t) \buildrel \Delta \over = {\rm{col}}\{ w(t),{v_1}(t+1),{v_2}(t+1), \cdots ,{v_L}(t+1)\}$, while ${A_F^L}(t)$, ${B_F^L}(t)$ and $\Omega (t)$ are defined in (\ref{eq:10}) and (\ref{eq:13}), respectively. Define $\bar \xi (t) \buildrel \Delta \over = {\rm{col}}\{ {{\rm{e}}_F}(t),\xi (t)\}$, then we introduce three matrices ${\Upsilon (t)}$, ${P(t)}>0$ and ${\Theta (t)}>0$ such that
\begin{eqnarray}
{\rm{e}}_0^{\rm{T}}(t){{\rm{e}}_0}(t) < {{\bar \xi }^{\rm{T}}}(t - 1)\left[ {\begin{array}{*{20}{c}}
   {P(t)} & {\Upsilon (t)}  \\
    *  & {\Theta (t)}  \\
\end{array}} \right]\bar \xi (t - 1)
\label {eq:23}
\end{eqnarray}
To guarantee that the right term in (\ref{eq:23}) is an upper bound of ${\rm{e}}_0^{\rm{T}}(t){{\rm{e}}_0}(t)$, the following inequality must be satisfied:
\begin{eqnarray}
{{\bar \xi }^{\rm{T}}}(t - 1)\underbrace {\left[ {\begin{array}{*{20}{c}}
   {{\Lambda _1}(t)} & {{\Lambda _2}(t)}  \\
    *  & {{\Lambda _3}(t)}  \\
\end{array}} \right]}_{\Lambda (t)}\bar \xi (t - 1) < 0
\label {eq:24}
\end{eqnarray}
where ${\Lambda _1}(t) \buildrel \Delta \over = {[A_F^L(t)]^{\rm{T}}}{\Omega ^{\rm{T}}}(t)\Omega (t){A_F^L}(t) - P(t)$, ${\Lambda _2}(t) \buildrel \Delta \over = {[A_F^L(t)]^{\rm{T}}}{\Omega ^{\rm{T}}}(t)\Omega (t){B_F^L}(t)-{\Upsilon (t)}$ and ${\Lambda _3}(t) \buildrel \Delta \over = {[B_F^L(t)]^{\rm{T}}}{\Omega ^{\rm{T}}}(t)\Omega (t){B_F^L}(t) - \Theta (t)$. Under this case, the first inequality in (\ref{eq:13}) is equivalent to $\Lambda (t) < 0$ according to the Schur complement lemma. Notice that ${\rm{Tr}}\left\{ {\left[ {\begin{array}{*{20}{c}}
   {P(t)} & {\Upsilon (t)}  \\
    *  & {\Theta (t)}  \\
\end{array}} \right]} \right\} = {\rm{Tr}}\{ P(t) + \Theta (t)\}$, and thus it follows from (\ref{eq:23}) that
\vspace{-6pt}
\begin{eqnarray}
{\rm{e}}_0^{\rm{T}}(t){{\rm{e}}_0}(t) < {\lambda _{\max }}(\bar \xi (t - 1){{\bar \xi }^{\rm{T}}}(t - 1)){\rm{Tr}}\{ P(t) + \Theta (t)\}\;\;
\label {eq:26}
\end{eqnarray}
Since $\bar \xi (t - 1)$ cannot be changed by the fusion system (\ref{eq:22}), $\min \;{\rm{Tr}}\{ P(t) + \Theta (t)\}$ can be chosen as the optimization objective when designing optimal weighting matrices. This completes the proof.
\end{proof}

Based on Theorem 1, the computation procedures for the DFE ${\rm{\hat x}}(t)$ of linear systems (\ref{eq:5}--\ref{eq:6}) are summarized as follows:
\begin{algorithm}
\caption{}
\begin{algorithmic}[1]\label{algo:1}
\STATE Determine the local estimator gains ${{\rm K}_i^L}(t)(i = 1, \cdots ,L)$ by solving the optimization problem (\ref{eq:11});
\STATE Determine the optimal weighting matrices ${\Omega _i}(t)(i = 1, \cdots ,L)$ by solving the optimization problem (\ref{eq:13});
\STATE Calculate linear LSEs ${\rm{\hat x}}_i(t)(i = 1, \cdots ,L)$ by (\ref{eq:7});
\STATE Calculate the DFE ${\rm{\hat x}}(t)$ by (\ref{eq:8});
\STATE Return to Step 1 and implement Steps 1--4 for calculating ${\rm{\hat x}}(t+1)$.
\end{algorithmic}
\end{algorithm}

\textbf{Remark 4}. Different from the constructing method of the upper bounds in \cite{c20}, the matrices ${\Theta _i}(t)$ in (\ref{eq:14}), $\Upsilon (t)$ and $\Theta (t)$ in (\ref{eq:23}) are introduced to construct the upper bounds of the SEs of the LSEs and DFE in this paper. Notice that the estimation performance of the LSEs and DFE is strongly dependent on the constructed upper bounds. When establishing the convex optimization problems in Theorem 1, these introduced matrices can increase the search space, and thus may lead to better solutions. On the other hand, the norm inequality was used in \cite{c20} to find an upper bound of the SEs and determine an optimization objective by a further inequality relaxation. Different from the relaxation inequality in \cite{c20}, an upper bound of the SEs in this paper is constructed by the trace inequalities (see (\ref{eq:18}) and (\ref{eq:26})), and the corresponding optimization objectives are also determined by the same trace inequalities. Notice that the constructed optimization objectives do not require further inequality relaxation. From the above analysis, the upper bound after optimizing the objective in Theorem 1 is closer to the real SE at each time. Therefore, the conservatism of the estimation method in Theorem 1 can be less than the estimation method in \cite{c20} because of these introduced new matrices, different relaxation inequalities and optimization objectives.
\vspace{-4pt}
\subsection{DFE Design for Nonlinear Systems}
Based on the DFE design method of linear systems in Subsection III.A, an optimal nonlinear estimator gain ${\rm K}_i^N(t)$ in (\ref{eq:4}) and optimal weighting matrices ${\Omega _1}(t), \cdots ,{\Omega _L}(t)$ in (\ref{eq:8}) for nonlinear systems (\ref{eq:1}--\ref{eq:2}) will be presented in Theorem 2.

\begin{theorem}
Define
\begin{eqnarray}
\left\{ {\begin{array}{*{20}{l}}
   {{A_{J_i}}(t - 1) = {{\left. {\frac{{\partial {\rm{f}}({\rm{x}}(t - 1))}}{{\partial {\rm{x}}(t - 1)}}} \right|}_{{\rm{x}}(t - 1) = {{{\rm{\hat x}}}_i}(t - 1)}}}  \\
   {{C_{{J_i}}}(t) = {{\left. {\frac{{\partial {{\rm{g}}_i}({\rm{x}}(t))}}{{\partial {\rm{x}}(t)}}} \right|}_{{\rm{x}}(t) = {\rm{f}}({{{\rm{\hat x}}}_i}(t - 1))}}}  \\
\end{array}} \right.
\label {eq:29}
\end{eqnarray}
An optimal estimator gain ${\rm K}_i^N(t)$ in (\ref{eq:4}) can be obtained by solving the following convex optimization problem:
\begin{eqnarray}
\begin{array}{l}
 \mathop {\min }\limits_{{\Pi _i}(t) > 0,{\Upsilon _i}(t) > 0,{{\rm M}_i}(t) > 0,{\eta _i}(t) > 0,{\Psi _i}(t),{\rm K}_i^N(t)} {\rm{Tr}}\{ {\Upsilon _i}(t) + {{\rm M}_i}(t)\}  \\
 {\rm{s}}{\rm{.t}}{\rm{.:}}\left\{ \begin{array}{l}
 \left[ {\begin{array}{*{20}{c}}
   { - I} & {{\rm{G}}_{{{\rm K}_i}}^N(t){A_{{J_i}}}(t - 1)} & {{{\rm X}_{i2}}(t)} & {{{\rm X}_{i3}}(t)}  \\
    *  & { - {\Pi _i}(t)} & 0 & 0  \\
    *  &  *  & { - {\Upsilon _i}(t)} & { - {\Psi _i}(t)}  \\
    *  &  *  &  *  & { - {{\rm M}_i}(t)}  \\
\end{array}} \right] < 0 \\
 {\Pi _i}(t) - {\eta _i}(t)I < 0 \\
 {\eta _i}(t) \le 1 \\
 \end{array} \right. \\
 \end{array}
\label {eq:30}
\end{eqnarray}
where
\begin{eqnarray}
\left\{ \begin{array}{l}
 {\rm{G}}_{{{\rm K}_i}}^N(t) = I - {\rm K}_i^N(t){C_{{J_i}}}(t) \\
 {{\rm X}_{i2}}(t) \buildrel \Delta \over = {\rm{G}}_{{{\rm K}_i}}^N(t)B(t - 1)\\
 {{\rm X}_{i3}}(t) \buildrel \Delta \over = - {\rm K}_i^N(t){B_i}(t)\\
 \end{array} \right.
\label {eq:31}
\end{eqnarray}
Then, a group of optimal weighting matrices ${\Omega _1}(t), \cdots ,{\Omega _L}(t)$ in (\ref{eq:8}) for nonlinear systems (\ref{eq:1}--\ref{eq:2}) can be obtained by solving the following convex optimization problem:
\begin{eqnarray}
\begin{array}{l}
 \mathop {\min }\limits_{\Pi (t) > 0,{\rm{M}}(t){\rm{ > 0,}}\Psi (t),\Omega (t)} {\rm{Tr}}\{ \Pi (t) + {\rm{M}}(t)\}  \\
 {\rm{s}}{\rm{.t}}{\rm{.}}:\left\{ {\left[ {\begin{array}{*{20}{c}}
   { - I} & {\Omega (t)A_F^N(t)} & {\Omega (t)B_F^N(t)}  \\
    *  & { - \Pi (t)} & { - \Psi (t)}  \\
    *  &  *  & { - {\rm{M}}(t)}  \\
\end{array}} \right] < 0} \right. \\
 \end{array}
\label {eq:32}
\end{eqnarray}
where ${A_F^N(t)}$, ${B_F^N(t)}$ and ${\Omega (t)}$ are defined by
\begin{eqnarray}
\left\{ \begin{array}{l}
 B_{{f_i}}^N(t) \buildrel \Delta \over = [{\rm{G}}_{{{\rm K}_i}}^N(t)B(t - 1)\;\; - {\rm K}_i^N(t){B_i}(t)] \\
 A_F^N(t) \buildrel \Delta \over = {\rm{diag}}\{ {\rm{G}}_{{{\rm K}_1}}^N(t){A_{{J_1}}}(t - 1), \\
 \;\;\;\;\;\;\;\;\;\;\;\;\;\;\;\;\;\;\;\;\;\;\;\;\;\;\;\;\;\;\;\;\;\;\;\;\;\cdots,{\rm{G}}_{{{\rm K}_L}}^N(t){A_{{J_L}}}(t - 1)\}  \\
 B_F^N(t) \buildrel \Delta \over = {\rm{diag}}\{ B_{{f_1}}^N(t), \cdots ,B_{{f_L}}^N(t)\}  \\
 \Omega (t) \buildrel \Delta \over = [{\Omega _1}(t), \cdots ,{\Omega _{L - 1}}(t),I - \sum\nolimits_{i = 1}^{L - 1} {{\Omega _i}(t)} ] \\
 \end{array} \right.
\label {eq:33}
\end{eqnarray}
\end{theorem}

\begin{proof}
Define ${\rm{\tilde x}}_i^{\rm{p}}(t)\; \buildrel \Delta \over = {\rm{f}}({\rm{x}}(t - 1)) - {\rm{f}}({{{\rm{\hat x}}}_i}(t - 1)) + B(t - 1)w(t - 1)$ and ${{{\rm{\tilde x}}}_i}(t) \buildrel \Delta \over = {\rm{x}}(t) - {{{\rm{\hat x}}}_i}(t)$. Then, it follows from (\ref{eq:4}) that
\begin{eqnarray}
\begin{array}{l}
 {{{\rm{\tilde x}}}_i}(t) = {\rm{\tilde x}}_i^{\rm{p}}(t) - {\rm K}_i^N(t)[{{\rm{g}}_i}({\rm{x}}(t)) - {{\rm{g}}_i}({\rm{\hat x}}_i^{\rm{p}}(t)) + {B_i}(t){v_i}(t)] \\
 \end{array}
\label {eq:35}
\end{eqnarray}
Meanwhile, by expanding ${\rm{f}}({\rm{x}}(t - 1))$ and ${{\rm{g}}_i}({\rm{x}}(t))$ in Taylor series about ``${{{\rm{\hat x}}}_i}(t - 1)$
'' and ``${\rm{\hat x}}_i^{\rm{p}}(t)$'', respectively, one has
\begin{eqnarray}
\left\{ \begin{array}{l}
 {\rm{f}}({\rm{x}}(t - 1)) = {\rm{f}}({{{\rm{\hat x}}}_i}(t - 1)) + {A_{J_i}}(t - 1){{{\rm{\tilde x}}}_i}(t - 1) \\
 \;\;\;\;\;\;\;\;\;\;\;\;\;\;\;\;\;\;\;\; + {\Delta _{\rm{f}}}({\rm{\tilde x}}_i^2(t - 1)) \\
 {{\rm{g}}_i}({\rm{x}}(t)) = {{\rm{g}}_i}({\rm{\hat x}}_i^{\rm{p}}(t)) + {C_{{J_i}}}(t){\rm{\tilde x}}_i^{\rm{p}}(t)  + {\Delta _{{{\rm{g}}_i}}}({[{\rm{\tilde x}}_i^{\rm{p}}(t)]^2}) \\
 \end{array} \right.
\label {eq:36}
\end{eqnarray}
where ${{A_{{J_i}}}(t - 1)}$ and ${{C_{{J_i}}}(t)}$ are defined by (\ref{eq:29}), while ${\Delta _{\rm{f}}}({\rm{\tilde x}}_i^2(t - 1))$ and ${\Delta _{{{\rm{g}}_i}}}({[{\rm{\tilde x}}_i^{\rm{p}}(t)]^2})$ represent the high-order terms of the Taylor series expansion.

According to (\ref{eq:36}), the nonlinear estimation error system (\ref{eq:35}) is equivalent to:
\begin{eqnarray}
\begin{array}{l}
 {{{\rm{\tilde x}}}_i}(t) = {\rm{G}}_{{{\rm K}_i}}^N(t){A_{J_i}}(t - 1){{{\rm{\tilde x}}}_i}(t - 1) \\
 \;\;\;\;\;\;\;\;\;\;\; + {\rm{G}}_{{{\rm K}_i}}^N(t)[B(t - 1)w(t - 1) + {\Delta _{\rm{f}}}({\rm{\tilde x}}_i^2(t - 1))] \\
 \;\;\;\;\;\;\;\;\;\;\; - {\rm K}_i^N(t)[{B_i}(t){v_i}(t) + {\Delta _{{{\rm{g}}_i}}}({[{\rm{\tilde x}}_i^{\rm{p}}(t)]^2})] \\
 \end{array}
\label {eq:37}
\end{eqnarray}
where ${\rm{G}}_{{{\rm K}_i}}^N(t)$ is defined by (\ref{eq:31}). Notice that ${\Delta _{\rm{f}}}({\rm{\tilde x}}_i^2(t - 1))$ and ${\Delta _{{{\rm{g}}_i}}}({[{\rm{\tilde x}}_i^{\rm{p}}(t)]^2})$ in (\ref{eq:37}) are unknown noises, while $w(t-1)$ and ${v_i}(t)$ are also unknown noises. Under this case, the terms $B(t - 1){{\tilde w}_i}(t - 1)$ and ${B_i}(t){{\tilde v}_i}(t)$ are introduced to model the affection factors caused by unknown noises. Then, (\ref{eq:37}) can be written as:
\begin{eqnarray}
\begin{array}{l}
 {{{\rm{\tilde x}}}_i}(t) = {\rm{G}}_{{{\rm K}_i}}^N(t){A_{{J_i}}}(t - 1){{{\rm{\tilde x}}}_i}(t - 1) \\
  \;\;\;\;\;\;\;\;+ {\rm{G}}_{{{\rm K}_i}}^N(t)[B(t - 1){{\tilde w}_i}(t - 1)] - {\rm K}_i^N(t)[{B_i}(t){{\tilde v}_i}(t)] \\
 \end{array}
\label {eq:38}
\end{eqnarray}
Since the form of (\ref{eq:38}) is the same as that of (\ref{eq:a6}), the optimization problems (\ref{eq:30}) and (\ref{eq:32}) in Theorem 2 can be obtained by the similar proof in Theorem 1, and the detailed derivation is omitted due to the page limitation.
\end{proof}

Based on Theorem 2, the computation procedures for the DFE ${\rm{\hat x}}(t)$ of nonlinear systems (\ref{eq:1}--\ref{eq:2}) are summarized as follows:
\begin{algorithm}
\caption{}
\begin{algorithmic}[1]\label{algo:1}
\STATE Calculate the matrices ${{A_{{J_i}}}(t - 1)}$ and ${{C_{{J_i}}}(t)}$ by (\ref{eq:29});
\STATE Determine the local estimator gains ${\rm K}_i^N(t)(i = 1, \cdots ,L)$ by solving the optimization problem (\ref{eq:30});
\STATE Determine the optimal weighting matrices ${\Omega _i}(t)(i = 1, \cdots ,L)$ by solving the optimization problem (\ref{eq:32});
\STATE Calculate nonlinear LSEs ${\rm{\hat x}}_i(t)(i = 1, \cdots ,L)$ by (\ref{eq:4});
\STATE Calculate the nonlinear DFE ${\rm{\hat x}}(t)$ by (\ref{eq:8});
\STATE Return to Step 1 and implement Steps 1--5 for calculating ${\rm{\hat x}}(t+1)$.
\end{algorithmic}
\end{algorithm}

\textbf{Remark 5}. The optimization problems (\ref{eq:11}), (\ref{eq:13}), (\ref{eq:30}) and (\ref{eq:32}) are established in terms of linear matrix inequalities, and thus they can be directly solved by the function ``\emph{mincx}'' of MATLAB LMI Toolbox \cite{c23}. On the other hand, it can be concluded from (\ref{eq:a9}) that when each linear LSE is designed by (\ref{eq:11}), the SE of the DFE for the linear time-varying systems (\ref{eq:5}--\ref{eq:6}) must be bounded at each time.

\section{Simulation Examples}
\subsection {Target Tracking System}
Consider a maneuvering target which is monitored by two sensors, and define the state vector $x(t)$ by $x(t) \buildrel \Delta \over = {\rm{col}}\{ s(t),\dot s(t)\}$, where $s(t)$ is the target's position, and $\dot s(t)$ is the target's velocity. Then, the dynamical process of the target's position and velocity can be modeled by\cite{c1}:
\begin{eqnarray}
x(t + 1) = \left[ {\begin{array}{*{20}{c}}
   1 & {f_s(t)}  \\
   0 & 1  \\
\end{array}} \right]x(t) + \left[ {\begin{array}{*{20}{c}}
   {0.5{f_s^2}(t)}  \\
   {f_s(t)}  \\
\end{array}} \right]w(t)
\label {eq:s1}
\end{eqnarray}
where $f_s(t)$ is the time-varying sampling period, and $w(t)$ is the process noise. Then, each sensor's measurement is modeled by
\begin{eqnarray}
{y_i}(t) = {C_i}(t)x(t) + {B_i}(t)v(t)
\label {eq:s2}
\end{eqnarray}
where $v(t)$ is the measurement noise, and ${C_1}(t) = [0.5\;\;1],{C_2}(t) = [1\;0]$, ${B_1}(t) = 1.2\cos (f_s(t)),{B_2}(t) = 2.0\sin (f_s(t))$. In the simulation, three types of disturbance noises $w(t)$ and $v(t)$ in (\ref{eq:s1}) and (\ref{eq:s2}) will be considered:
\begin{itemize}
\item \textbf{Type I:} When $f_s(t) = 0.5$, $w(t)$ and $v(t)$ are the \emph{energy-bounded noises} given by $w(t) = (2 + 0.2\cos (t))\exp ( - t/9)$ and $v(t)= 0.8\sin (t)\exp ( - t/6)$.

\item \textbf{Type II:} When $f_s(t) = 0.5 + 0.2\sin (t)$, $w(t)$ and $v(t)$ are the uncorrected \emph{Gaussian white noises} with covariances ${Q_w} = 1.8$ and ${Q_v} = 0.5$;
\item \textbf{Type III:} When $f_s(t) = 0.5 + 0.2\sin (t)$, $w(t)$ and $v(t)$ are the \emph{bounded noises} given by $w(t) = \cos (t) - 0.5$ and $v(t) = 0.7\sin (t) - 0.3$.
\end{itemize}
Notice that, for the system (\ref{eq:s1}--\ref{eq:s2}), each LSE can be obtained by the ${H_\infty }$ filter in \cite{c22} under Type I, the well-known Kalman filtering method (see \cite{c21}) can be used to design each LSE under Type II, and each LSE can be obtained by Theorem 1 in \cite{c20} under Type III.
\begin{figure}[thpb]
\centering
\includegraphics[height=6.5cm, width=8.5cm]{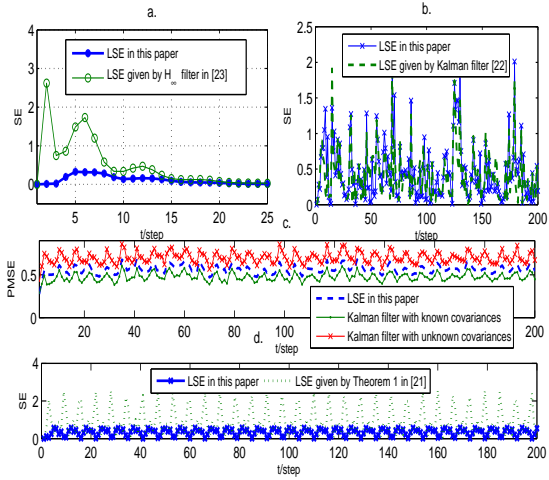}
 \caption{According to the first sensor measurements, the corresponding LSE is designed by using different methods: (a): Under the Type I, SEs of the ${H_\infty }$ filter in \cite{c22} and the LSE in this paper; (b): Under the Type II, SEs of the Kalman filter \cite{c21} and the LSE in this paper; (c): Under the Type II, PMSEs of the LSE in this paper and the Kalman filters with known covariances and unknown covariances; (d): Under the Type III, SEs of the LSE in \cite{c20} and the LSE in this paper.}
 \label{fig1}
 \end{figure}

To demonstrate the advantages of the designed estimation algorithm, the estimation performances of the first LSEs are shown in Fig.1 by using different estimation methods. Since energy-bounded noise is a special case of bounded noises, the LSE under Type I can also be obtained by Theorem 1 in this paper. It is seen from Fig.1(a) that the estimation precision of the LSE in this paper is higher than that of the ${H_\infty }$ filtering in \cite{c22}. Meanwhile, the Gaussian white noise is always bounded in a practical system, thus it can be viewed as a class of bounded noises. Under this case, the designed LSE in this paper can be applicable to Type II. Then, it is seen from Fig.1(b) that the estimation precision of the LSE given by Theorem 1 is close to that of the Kalman filter \cite{c21}. Due to the random noises, the estimation performance of the LSE is assessed by its mean square error (MSE), and the Monte Carlo method is used to approach the theory MSE. Then, the practical MSEs (PMSEs) of the LSE in this paper and the Kalman filter with accurate/inaccurate covariances are shown in Fig.1(c). It is seen from this figure that when the covariances are known, the estimation performance of Kalman filter is better than that of the LSE in this paper. This is because the Kalman filter is designed in the linear minimum variance sense at each time step. On the other hand, when the covariances are unknown or inaccurately known, the estimation precision of the LSE in this paper is higher than that of the Kalman filter. This is because the Kalman filter is required to know the covariance. In contrast, the LSE designed by this paper is not required to know the statistical information of noises. The above discussion implies that the designed LSE in this paper is more robust, and is applicable to a more general case.

Furthermore, when considering the bounded noises, Fig.1(d) shows that the estimation performance of the LSE in this paper is better than that of the LSE in \cite{c20}. This is because more available information is used to design estimator under Kalman-like structure (see Remark 1), and new upper bound and optimization objective constructed in Theorem 1 can reduce the conservatism (see Remark 4). Moreover, when considering the second measurement equation (i.e., $C_2(t)=[1\;0]$), the optimization problem in Theorem 1 of \cite{c20} is unsolvable.

\subsection{Mobile Robot Localization}
Consider the localization of an unicycle mobile robot operating in planar environments. Let ${p_r}(t) \buildrel \Delta \over = {\rm{col}}\{ {s_x}(t),{s_y}(t)\}$ denote the robot's position in X-Y plan, while $\theta(t)$ is to define the angular orientation. Then, the motion model of the robot is given by\cite{c37}:
\begin{eqnarray}
\left\{ \begin{array}{l}
 {s_x}(t + 1) = {s_x}(t) - \frac{{{{\hat u}_p}}}{{{{\hat u}_r}}}\left( {\sin \theta (t) - \sin (\theta (t) + {T_0}{{\hat u}_r})} \right) \\
 {s_y}(t + 1) = {s_y}(t) + \frac{{{{\hat u}_p}}}{{{{\hat u}_r}}}\left( {\cos \theta (t) - \cos (\theta (t) + {T_0}{{\hat u}_r})} \right) \\
 \theta (t + 1) = \theta (t) + {T_0}{{\hat u}_r} + {T_0}{w_\theta }(t) \\
 {{\hat u}_p} = {u_p} + {w_p}(t) \\
 {{\hat u}_r} = {u_r} + {w_r}(t) \\
 \end{array} \right.
\label {eq:s4}
\end{eqnarray}
where ${T_0}$ is the sampling period, ${w_\theta }(t)$ is the additional rotational noise; ${u_p}$ is the motion command to control the translational velocity, while ${u_r}$ is the motion command to control the rotational velocity. As pointed out in \cite{c37}, \emph{robot motion is subject to noise in reality}, i.e., the motion commands ${u_p}$ and ${u_r}$ may be changed by the unpredictable disturbances. Then, the true velocity control input ${{\hat u}_p}$ (or ${{\hat u}_r}$) equals the commanded velocity plus some small, additive noise ${w_p}(t)$ (or ${w_r}(t)$). Notice that the motion commands ${u_p}$ and ${u_r}$ are constant and known for the robot model (\ref{eq:s4}), and thus the model (\ref{eq:s4}) can be written as:
\begin{eqnarray}
{\rm{x}}(t + 1) = {\rm{f}}({\rm{x}}(t)) + \Gamma w(t)
\label {eq:s5}
\end{eqnarray}
where ${\rm{x}}(t) \buildrel \Delta \over = {\rm{col}}\{ {s_x}(t),{s_y}(t),\theta (t)\}$, $\Gamma  \buildrel \Delta \over = {\rm{diag}}\{ 1,1,{T_0}\}$, $w(t) \buildrel \Delta \over = {\rm{col}}\{ {w_1}(t),{w_2}(t),{w_3}(t)\}$, and
\begin{eqnarray}
\left\{ \begin{array}{l}
 {\rm{f}}({\rm{x}}(t)) = \left[ \begin{array}{l}
 {s_x}(t) - \frac{{{u_p}}}{{{u_r}}}\left( {\sin \theta (t) - \sin (\theta (t) + {T_0}{u_r})} \right) \\
 {s_y}(t) + \frac{{{u_p}}}{{{u_r}}}\left( {\cos \theta (t) - \cos (\theta (t) + {T_0}{u_r})} \right) \\
 \theta (t) + {T_0}{u_r} \\
 \end{array} \right] \\
 {w_1}(t) = \left( {\frac{{{u_p}}}{{{u_r}}} - \frac{{{{\hat u}_p}}}{{{{\hat u}_r}}}} \right)\sin \theta (t) + \frac{{{{\hat u}_p}}}{{{{\hat u}_r}}}\sin (\theta (t) + {T_0}{{\hat u}_r}) \\
 \;\;\;\;\;\;\;\;\; - \frac{{{u_p}}}{{{u_r}}}\sin (\theta (t) + {T_0}{u_r}) \\
 {w_2}(t) = \left( {\frac{{{{\hat u}_p}}}{{{{\hat u}_r}}} - \frac{{{u_p}}}{{{u_r}}}} \right)\cos \theta (t) + \frac{{{u_p}}}{{{u_r}}}\cos (\theta (t) + {T_0}{u_r}) \\
 \;\;\;\;\;\;\;\;\; - \frac{{{{\hat u}_p}}}{{{{\hat u}_r}}}\cos (\theta (t) + {T_0}{{\hat u}_r}) \\
 {w_3}(t) = {w_r}(t) + {w_\theta }(t) \\
 \end{array} \right.
\label {eq:s6}
\end{eqnarray}
Since ${w_1}(t)$ and ${w_2}(t)$ are dependent on the state $\theta (t)$, they are state-dependent noises for the system (\ref{eq:s5}).

In X-Y plan, it is considered that four known points, denoted as $({s_{{x_i}}},{s_{{y_i}}})(i = 1,2,3,4)$, are chosen as the landmarks. Then, the distance from the robot's planner Cartesian coordinates $({s_x}(t),{s_y}(t))$ to each landmark $({s_{{x_i}}},{s_{{y_i}}})$ can be expressed as follows:
\begin{eqnarray}
{d_i}(t) = \sqrt {{{({s_{{x_i}}} - {s_x}(t))}^2} + {{({s_{{y_i}}} - {s_y}(t))}^2}}
\label {eq:s7}
\end{eqnarray}
The azimuth ${\varphi _i}(t)$ at time $t$ can be related to the current system state variables ${s_x}(t)$, ${s_y}(t)$ and $\theta (t)$ as follows:
\begin{eqnarray}
{\varphi _i}(t) = \theta (t) - \arctan \left( {\frac{{{s_{{y_i}}} - {s_y}(t)}}{{{s_{{x_i}}} - {s_x}(t)}}} \right)
\label {eq:s8}
\end{eqnarray}
Both the distance ${d_i}(t)$ and ${\varphi _i}(t)$ are treated as the measurements. Furthermore, when considering the unpredicted disturbances, the measurement equations for the robotic system (\ref{eq:s4}) can be written as follows:
\begin{eqnarray}
\left\{ \begin{array}{l}
 {y_1}(t) = \left[ {\begin{array}{*{20}{c}}
   {{{\rm{g}}_1}({\rm{x}}(t))}  \\
   {{{\rm{g}}_2}({\rm{x}}(t))}  \\
\end{array}} \right] + \left[ {\begin{array}{*{20}{c}}
   {{D_1}} & 0  \\
   0 & {{D_2}}  \\
\end{array}} \right]\left[ {\begin{array}{*{20}{c}}
   {{v_1}(t)}  \\
   {{v_2}(t)}  \\
\end{array}} \right] \\
 {y_2}(t) = \left[ {\begin{array}{*{20}{c}}
   {{{\rm{g}}_3}({\rm{x}}(t))}  \\
   {{{\rm{g}}_4}({\rm{x}}(t))}  \\
\end{array}} \right] + \left[ {\begin{array}{*{20}{c}}
   {{D_3}}  \\
   {{D_4}}  \\
\end{array}} \right]{v_3}(t) \\
 \end{array} \right.
\label {eq:s9}
\end{eqnarray}
where $v_1(t)$, $v_2(t)$ and $v_3(t)$ are the measurement noises, and
\begin{eqnarray}
\left\{ \begin{array}{l}
 {{\rm{g}}_i}({\rm{x}}(t)) = {\rm{col}}\{ {d_i}(t),{\varphi _i}(t)\} (i = 1,2,3,4) \\
 {D_1} = {\rm{diag}}\{ 0.5,0.3\} ,{D_2} = {\rm{diag}}\{ 0.3,0.5\}  \\
 {D_3} = {\rm{diag}}\{ 0.2,0.6\} ,{D_4} = {\rm{diag}}\{ 0.5,0.7\}  \\
 \end{array} \right.
\label {eq:s10}
\end{eqnarray}
Then, based on the sensor measurements (\ref{eq:s9}), the localization of this mobile robot can be realized by using different nonlinear estimation methods. On the other hand, by using Taylor series expansions, the linearized matrices ${A_{J_j}}(t)$ and $C_{J_j}(t)(j=1,2)$ for the nonlinear vector functions ${\rm{f}}({\rm{x}}(t))$ in (\ref{eq:s5}) and ${{\rm{g}}_i}({\rm{x}}(t))$ in (\ref{eq:s9}) near the point ${x^ * } \in {{\rm{R}}^3}$ can be expressed as follows:
\begin{eqnarray}
\left\{ \begin{array}{l}
 {A_{{J_j}}}(t) = {\left[ {\begin{array}{*{20}{c}}
   1 & 0 & { - \frac{{{u_p}}}{{{u_r}}}\cos {\theta (t)} + \frac{{{u_p}}}{{{u_r}}}\cos ({\theta (t)} + {T_0}{u_r})}  \\
   0 & 1 & { - \frac{{{u_p}}}{{{u_r}}}\sin {\theta (t)} + \frac{{{u_p}}}{{{u_r}}}\sin ({\theta (t)} + {T_0}{u_r})}  \\
   0 & 0 & 1  \\
\end{array}} \right]_{{\rm{x}}(t) = {{\rm{x}}^ * }}} \\
 {D_{{J_i}}}(t) = {\left[ {\begin{array}{*{20}{c}}
   {\frac{{ - {{\tilde s}_{{x_i}}}(t)}}{{\sqrt {\tilde s_{{x_i}}^2(t) + \tilde s_{{y_i}}^2(t)} }}} & {\frac{{ - {{\tilde s}_{{y_i}}}(t)}}{{\sqrt {\tilde s_{{x_i}}^2(t) + \tilde s_{{y_i}}^2(t)} }}} & 0  \\
   {\frac{{ - {{\tilde s}_{{y_i}}}(t)}}{{\tilde s_{{x_i}}^2(t) + \tilde s_{{y_i}}^2(t)}}} & {\frac{{{{\tilde s}_{{x_i}}}(t)}}{{\tilde s_{{x_i}}^2(t) + \tilde s_{{y_i}}^2(t)}}} & 1  \\
\end{array}} \right]_{{\rm{x}}(t) = {{\rm{x}}^ * }}} \\
 {C_{{J_1}}}(t) \buildrel \Delta \over = {\rm{col}}\{ {D_{{J_1}}}(t),{D_{{J_2}}}(t)\}  \\
 {C_{{J_2}}}(t) \buildrel \Delta \over = {\rm{col}}\{ {D_{{J_3}}}(t),{D_{{J_4}}}(t)\}  \\
 \end{array} \right.
\label {eq:s11}
\end{eqnarray}
where $i=1,2,3,4$, ${{\tilde s}_{{x_i}}}(t) \buildrel \Delta \over = {s_{{x_i}}} - {s_x}(t)$ and ${{\tilde s}_{{y_i}}}(t) \buildrel \Delta \over = {s_{{y_i}}} - {s_y}(t)$.

In the simulation, the parameters ${T_0}$, ${u_p}$ and ${u_r}$ are taken as ${T_0} = 1$, ${u_p} = 0.075$ and ${u_r} = 0.025$, while four landmarks' positions in X-Y plan are $({s_{{x_1}}},{s_{{y_1}}}) = (5,10)$, $({s_{{x_2}}},{s_{{y_2}}}) = (10,10)$, $({s_{{x_3}}},{s_{{y_3}}}) = (10,5)$ and $({s_{{x_4}}},{s_{{y_4}}}) = (5,5)$. Meanwhile, it is reasonably considered that the initial robot's pose is known in advance, hence the initial estimation errors ${{{\rm{\tilde x}}}_1}(0)$ and ${{{\rm{\tilde x}}}_2}(0)$ can be given by ${{{\rm{\tilde x}}}_1}(0) = {{{\rm{\tilde x}}}_2}(0) = 0$. Furthermore, the disturbance noises ${w_p}(t)$, ${w_r}(t)$, ${w_\theta }(t)$ and $v_k(t)\;(k=1,2,3)$ in (\ref{eq:s4}--\ref{eq:s5}) are taken as follows:
\begin{itemize}
\item \textbf{Type IV:} ${w_p}(t)$, ${w_r}(t)$, ${w_\theta }(t)$ and $v_k(t)\;(k=1,2,3)$ are the bounded noises given by
\begin{eqnarray}
\left\{ \begin{array}{l}
 {w_p}(t) = 0.2{\rho _p}(t) - 0.1,{w_r}(t) = 0.3{\rho _r}(t) - 0.1 \\
 {w_\theta }(t) = 0.2{\rho _\theta }(t) - 0.1 \\
 {v_1}(t) = {\rm{col}}\{ 0.05{\rho _{{v_1}}}(t) - 0.01,0.02{\rho _{{v_2}}}(t) - 0.01\}  \\
 {v_2}(t) = {\rm{col}}\{ 0.03{\rho _{{v_3}}}(t) - 0.01,0.05{\rho _{{v_4}}}(t) - 0.03\}  \\
 {v_3}(t) = {\rm{col}}\{ 0.02{\rho _{{v_5}}}(t) - 0.01,0.06{\rho _{{v_6}}}(t) - 0.02\}  \\
 \end{array} \right.\;\;\;
\end{eqnarray}
where ${\rho _p}(t)( \in [0,1])$, ${\rho _r}(t)( \in [0,1])$, ${\rho _\theta }(t)( \in [0,1])$ and
${\rho _{{v_l}}}(t)( \in [0,1])\;(l=1,2,3,4,5,6)$ are random variables that can be generated by the function ``rand'' of MATLAB. Under this case, $w(t)$ in (\ref{eq:s5}) is also the bounded noise.
\end{itemize}

To demonstrate the effectiveness, by implementing Algorithm 2,
the actual robot trajectory in the X-Y plan and its position estimation are plotted in Fig. 2(a) under Type IV. It is seen from this figures that the mobile robot can get its position well by using the nonlinear fusion estimation algorithm in this paper. Meanwhile, the SEs of the DFE and LSEs are depicted in Fig. 2(c) under Type IV, which shows that the estimation performance of the DFE is better than that of each LSE. This is as expected for the fusion estimation methods. On the other hand, when the optimization problem (\ref{eq:30}) is solvable at each time, there must be ${J_{{d_j}}}(t) \buildrel \Delta \over = ||{\rm{G}}_{{{\rm K}_j}}^N(t){A_{{J_j}}}(t - 1)|{|_2} < 1$, which has been illustrated by Fig. 2(b).

\begin{figure}[thpb]
\centering
\includegraphics[height=6.5cm, width=8.5cm]{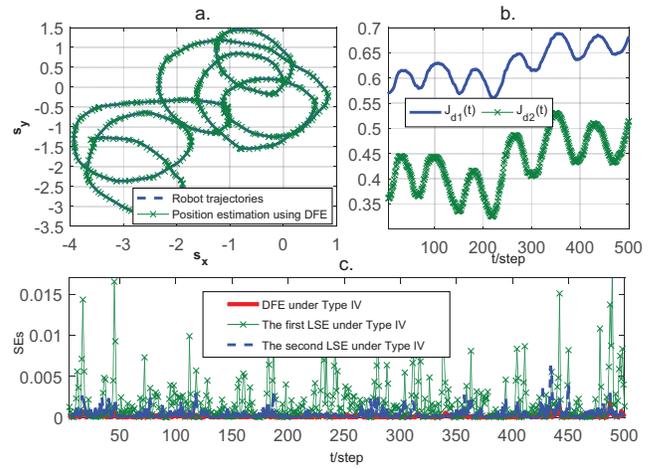}
 \caption{Nonlinear fusion estimation under Type IV: (a) The robot's trajectory and its position estimation by using Algorithm 2; (b) Trajectories of ${J_{{d_j}}}(t)(j=1,2)$; (c) Comparison of the estimation precision for the DFE and LSEs under Type IV.}
 \label{fig1}
 \end{figure}
\begin{figure}[thpb]
\centering
\includegraphics[height=6.0cm, width=8.5cm]{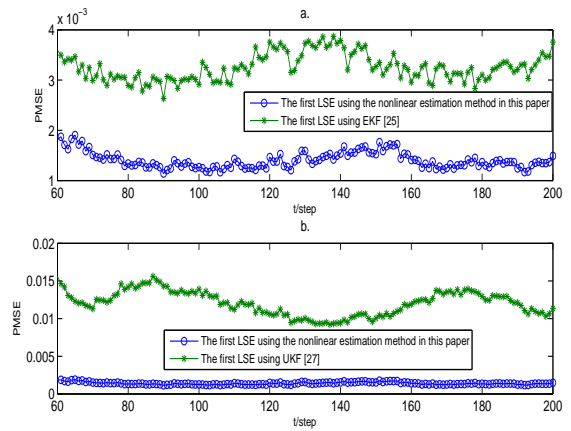}
 \caption{Under Type IV, the comparison of the first LSE's estimation performance by using the estimation method of this paper, EKF method\cite{c26} and UKF method\cite{c28}, respectively.}
 \label{fig1}
\end{figure}

To demonstrate the advantages of the developed nonlinear estimation method, it will be compared with the classical extended kalman filter (EKF) method \cite{c26} and unscented kalman filter (UKF) method \cite{c28,c29}, where these methods are all applicable to nonlinear systems with \emph{Gaussian white noises}, and the UKF  in \cite{c28} can also deal with state-dependent noises by using augmentation strategy. Here, the performance of the first LSE is used to show the advantages of the proposed method. Due to the random noises in Type IV, the estimation performance of the LSE is assessed by the practical mean square errors (PMSEs) that are calculated by Monte Carlo method \cite{c1} with an average of 500 runs. Then, the PMSEs of the first LSE using different methods are plotted in Fig. 3. It is seen from this figure that the estimation precision of the proposed method is higher than that of EKF or UKF method. This is because the disturbance noises are bounded, and no statistical information is available, but the EKF and UKF are all used to deal with the case of Gaussian white noises with known covariances. Notice that, as compared with Gaussian white noise with known covariance, the condition ``bounded noise" is easier to be satisfied in a practical system. This implies that the nonlinear estimation method in this paper is more robust as compared with the classic EKF and UKF methods.

\section{CONCLUSIONS}
In this paper, a new method to distributed fusion estimation problem has been developed for linear time-varying systems and nonlinear systems with bounded noises. When considering linear time-varying fusion systems, each local Kalman-like estimator with time-varying gains was designed such that the SE of each LSE must be bounded as time goes to $\infty $, and a novel distributed fusion estimation criterion was designed by establishing a convex optimization problem. Furthermore, the general nonlinear systems were transformed to linear time-varying systems by using Taylor series expansion, and the linearized errors could be viewed as unknown but bounded noises. Under this case, different convex optimization problems on the designs of the nonlinear estimator and distributed fusion criterion were established in terms of linear matrix inequalities for nonlinear fusion systems with bounded noises. Moreover, the solutions to the convex optimization problems in this paper can be directly obtained by using the Matlab LMI Toolbox. Finally, two illustrative examples were exploited to demonstrate the advantages and effectiveness of the proposed fusion estimation methods.

More recently, a great deal of attention has focused on the networked multi-sensor fusion estimation problem, where the sensor messages are transmitted to the fusion center through communication networks \cite{c3,c18,c20}. Therefore, when considering the communication uncertainties including bandwidth constraints, transmission delays and packet dropouts, one of our future works will focus on how to design the networked nonlinear fusion estimation algorithms based on the developed fusion estimation method in this paper.

%%%%%%%%%%%%%%%%%%%%%%%%%%%%%%%%%%%%%%%%%%%%%%%%%%%%%%%%%%%%%%%%%%%%%%%%%%%%%%%%

\end{document}